\documentclass[a4, 12pt, reqno]{article}
\usepackage{amscd,enumitem, amsmath, amssymb, amsfonts, amsthm}
\usepackage[arrow, matrix]{xy}
\usepackage{graphicx}
\usepackage{color}
\usepackage[utf8]{inputenc}
\usepackage[english]{babel}
\theoremstyle{plain}
\newtheorem{Theorem}{Theorem}[section]
\newtheorem{Lemma}{Lemma}[section]

\theoremstyle{definition}
\newtheorem{Problem}{Problem}

\numberwithin{equation}{section}
\begin{document}

\title{Solution to an open problem on the closeness of graphs}

\author{Fazal Hayat,  Shou-Jun Xu\footnote{Corresponding author.
 \newline E-mails: fhayatmaths@gmail.com (F. Hayat), shjxu@lzu.edu.cn (S.-J. Xu)} \\
School of Mathematics and Statistics,  Gansu Center for Applied Mathematics,\\
 Lanzhou University,  Lanzhou 730000,  P.R. China }

 \date{}
\maketitle

\begin{abstract}
 A network can be analyzed by means of many graph theoretical parameters. In the context of networks analysis,  closeness is a structural metric that evaluates a node's significance inside a network.   A cactus is a connected graph in which any block is either a cut edge or a cycle. This paper analyzes the closeness of cacti, we determine the unique graph that minimizes the closeness over all cacti with fixed numbers of vertices and cycles, which solves an open problem proposed by Poklukar \& \v{Z}erovnik  [Fundam. Inform. 167 (2019) 219--234]. \\\\
 {\bf Key words}: Closeness, residual closeness,  cactus, extremal graph.\\\\
{\bf 2010 Mathematics Subject Classification:} 68R10; 05C12; 05C35
\end{abstract}

\section{Introduction}

A network is usually represented by a simple graph where the vertices represent the nodes of the network and the edges represent the links between a pair of nodes that enable mutual communication. The key point in network analysis is the study and determination of their strength and fragility.  Graph theory has become one of the most powerful mathematical tools in network analysis and has many solution techniques and approaches in this regard. One of the most important tasks of network analysis is to determine which nodes or links are more critical in a network.   There are many graph-theoretic parameters related to the analysis of networks. Among these parameters, closeness is a way of detecting nodes that can spread information very efficiently through a  network. The closeness parameter is used in various scientific fields, for example,  in an information network,  closeness is a useful measure that estimates how fast the flow of information would be from a given node to other nodes. In social network analysis, it is used to find the individuals who are best placed to influence the entire network most quickly.

The initial closeness concept was first introduced by Freeman \cite{Free} in 1979. However, this parameter is unsuitable for disconnected graphs and also has some weaknesses in formulation and computation. To overcome the first deficiency, Latora and Marchiori \cite{LM} introduced a new measure of closeness  for disconnected graphs, which are still subject to the second weakness. Finally in 2006,  Danglachev \cite{Dan} proposed a rather different definition of the closeness, which can be used effectively for disconnected graphs and allows to creation of convenient formulae for graph computation. Based on the definition of closeness, many vulnerability measures have been defined to describe the resistance of a network. Among  these new measures, vertex (resp. edge) residual closeness parameters evaluate the closeness of a graph after vertex (resp. edge) removal \cite{Dan}. Another measure is additional closeness, which determines the maximal potential of the graph’s closeness via adding an edge \cite{Dan4, Dan5}. We refer the interested readers to \cite{AB, AO, CZ, Dan2,  LXZ, WZ, ZLG} for more detailed information about these new sensitive parameters.

Since its introduction, the closeness parameter has already incited a lot of research \cite{AT, Dan1, OA}. For example, Danglachev  obtained the closeness values for various graphs like complete graphs, stars, paths and  cycles in \cite{Dan},  splitting graphs in \cite{Dan3},   line graphs in \cite{Dan6}. Golpek \cite{Go} obtained  the closeness formulas for some classes of graphs like k-ary trees, binomial trees, binary trees, comet graphs, double comet graphs and double star graphs.  Poklukar and \v{Z}erovnik \cite{PZ} determined  the graphs that minimize or maximize the  closeness  among all graphs and several subclasses of graphs, including trees and cacti.  Recently, Zheng and Zhou \cite{ZZ} merged the concept of closeness in spectral graph theory, they determined the closeness matrix and established the relationship between the closeness eigenvalues and the structure of a graph.

Let $G $ be a simple connected  graph with vertex set $V(G)$ and edge set $E(G)$. For  $v \in V(G)$,  $N_G(v)$ denotes the set of vertices that are adjacent to $v$ in $G$. The degree of $v \in V(G)$, denoted by $d_G(v)$,  is the cardinality of  $N_G(v)$.
For $e \in E(G)$, let $G - e$ be the subgraph of $G$ obtained by deleting $e$ and  $G + xy$ be a graph obtained from $G$ by adding an edge connecting $x, y$, where  $x,y \in V(G)$. The graph formed from $G$ by deleting a vertex $v \in V(G)$ (and its incident edges) is denoted by $G - v$.
For  $u, v \in V(G)$, the  distance between $u$ and $v$ in $G$ is the least length of the path connecting $u$ and $v$ and is denoted by $d_G(u,v)$.  By $P_n$ and $C_n$ we denote the path and cycle  on $n$ vertices, respectively.

For a vertex $u$ of $G$, the closeness  of $u$ in $G$ is defined as \cite{Dan}
\[
 C_G(u)=  \sum_{v \in V(G)\setminus \{u\}} 2^{-d_G(u,v)}.
\]
The closeness of $G$ is defined as
\[
 C(G)= \sum_{u \in V(G) }C_G(u) = \sum_{u \in V(G) }\sum_{v \in V(G)\setminus \{u\}} 2^{-d_G(u,v)}.
\]

A cactus is a connected graph in which any block is either a cut edge or a cycle, or equivalently, a graph in which any two cycles have at most one common vertex. A cycle in a cactus is called an end-block if exactly one vertex of this cycle has a degree greater than  2. Let $\mathcal{D}(n, k)$ be the set of all cacti of order $n$ with $k$ cycles, where $0 \leq k \leq \lfloor \frac{n-1}{2} \rfloor$. The cactus graph has many applications in real life problems \cite{Das, Zm}.
Poklukar and \v{Z}erovnik \cite{PZ} determined  the unique graph that maximizes the closeness in $\mathcal{D}(n, k)$, and posed  an open problem for the minimum case.

\begin{Problem}\cite{PZ}\label{1}
 Find the unique graph that minimizes the closeness in  $\mathcal{D}(n, k)$, $0 \leq k \leq \lfloor \frac{n-1}{2} \rfloor$.
\end{Problem}

Let $k_1, k_2$ be two non-negative integers such that $k_1 + k_2 =k$, and let $D(n; k_1, k_2)$ be obtained from the path $P_{n-k}=u_1u_2 \cdots u_{n-k}$ and $k$ isolated vertices $\{v_1, v_2, \dots, v_k\}$ by adding $2k$ edges such that $N_G(v_i)= \{u_i, u_{i+1}\}$ for  $1 \leq i \leq k_1$ and  $N_G(v_i)= \{u_{n-2k-1+i}, u_{n-2k+i}\}$ for  $k_1+1 \leq i \leq k$. That is, $D(n; k_1, k_2)$ be the cactus obtained from $P_{n-k}=u_1u_2 \cdots u_{n-k}$ by replacing its first $k_1$ and last $k_2$ consecutive edges with $k_1$ and $k_2$ triangles, respectively.

In this paper, we identify the graph that minimizes the closeness among $\mathcal{D}(n, k)$, which contribute to the above mentioned Problem \ref{1} by proving the following Theorem.

\begin{Theorem}\label{a}
  Let $G$ be a graph in $\mathcal{D}(n, k)$,  where $0 \leq k \leq \lfloor \frac{n-1}{2} \rfloor$. Then $ C(G) \leq C (D(n; \lfloor \frac{k}{2} \rfloor, \lceil\frac{k}{2}\rceil))$
with equality  if and only if $G \cong D(n; \lfloor \frac{k}{2} \rfloor, \lceil\frac{k}{2}\rceil)$.
  \end{Theorem}

In Section 2, we give some transformations that minimize $C(G)$. As an application of these transformations,  in Section 3, we give the proof of Theorem \ref{a}.

\section{Preliminaries}

In this section, we propose some transformations that minimize $C(G)$.

\begin{Lemma} \cite{Dan4} \label{111}
For $i=1,2 $, let $H_i$ be a  graph with $v_i \in V(H_i)$. Let $G$ be the graph  obtained from $H_1$ and $ H_2$  by identifying $v_1$ with $v_2$. Then
$$
C({G})= C(H_1)+C(H_2)+2C_{H_1}(v_1)C_{H_2}(v_2).
$$
 \end{Lemma}

\begin{Lemma}\label{L3}
For $i=1,2,3$, let $H_i$ be a graph with $v_i \in V(H_i)$.  Let $u_1$ (resp. $u_2$) be the farthest vertex from $v_1$ (resp. $v_2$) in $H_1$ (resp. $H_2$) satisfying $C_{H_i}(u_i) \leq min _{1\leq i\leq 2}\{C_{H_i}(v_i)\}$. Let $G$ be the graph obtained from $H_1, H_2$ and $H_3$ by identifying $v_1, v_2$ and $v_3$ as a new vertex $v$. Let $G_1= G- \sum _{w \in N_{H_3}(v_3)}wv_3 + \sum _{w \in N_{H_3}(v_3)}wu_1$ and $G_2= G- \sum _{w \in N_{H_3}(v_3)}wv_3 + \sum _{w \in N_{H_3}(u_3)}wu_2$. Then either $G_1 < G$ or $  G_2 < G $.
\end{Lemma}
\begin{proof}
For convenience, let $Q= G - V(H_3) \setminus \{v_3\}$, $Q_1= G_1 - V(H_3) \setminus \{u_1\}$ and $Q_2= G_2 - V(H_3) \setminus \{u_2\}$. It is easy to see that $Q \cong Q_1 \cong Q_2$ and $C_{H_3}(v_3)=C_{H_3}(u_1)=C_{H_3}(u_2)$.
By Lemma \ref{111}, we have
\begin{equation*}\label{}
C(G)= C(H_3)+C(Q)+2C_{H_3}(v_3)C_{Q}(v),
\end{equation*}
and
\begin{equation*}\label{}
C(G_1)= C(H_3)+C(Q)+2C_{H_3}(v_3)C_{Q}(u_1).
\end{equation*}
This gives
\begin{equation*}\label{e1}
C(G_1)- C(G) = 2C_{H_3}(v_3)\left[C_{Q}(u_1)- C_{Q}(v)\right]
\end{equation*}
and
\begin{eqnarray*}\label{}
 &  & C_{Q}(u_1)- C_{Q}(v)   =   \sum_{x \in V(Q)\setminus \{u_1\}} 2^{-d_Q(u_1, x)} - \sum_{x \in V(Q)\setminus \{v\}} 2^{-d_Q(v, x)} \\
   & = &  \sum_{x \in V(H_1)\setminus \{u_1\}} 2^{-d_{H_1}(u_1, x)} + \sum_{x \in V(H_2)\setminus \{v_2\}} 2^{-(d_{H_2}(v_2, x) +  d_{H_1}(v_1, u_1))} \\
   &\quad -&  \sum_{x \in V(H_1)\setminus \{v_1\}} 2^{-d_{H_1}(v_1, x)}  -   \sum_{x \in V(H_2)\setminus \{v_2\}} 2^{-d_{H_2}(v_2, x)}\\
 &\leq  & \sum_{x \in V(H_2)\setminus \{v_2\}} 2^{-(d_{H_2}(v_2, x) +  d_{H_1}(v_1, u_1))}-\sum_{x \in V(H_2)\setminus \{v_2\}} 2^{-d_{H_2}(v_2, x)}<0.
\end{eqnarray*}
Hence, $C(G_1)< C(G)$.
Also, we have
\begin{equation*}\label{}
C(G_2)= C(H_3)+C(Q)+2C_{H_3}(v_3)C_{Q}(u_2).
\end{equation*}
We get
\begin{equation*}\label{e2}
C(G_2)- C(G) = 2C_{H_3}(v_3)\left[C_{Q}(u_2)- C_{Q}(v)\right]
\end{equation*}
and
\begin{eqnarray*}\label{}
 &  & C_{Q}(u_2)- C_{Q}(v)   =   \sum_{x \in V(Q)\setminus \{u_2\}} 2^{-d_Q(u_2, x)} - \sum_{x \in V(Q)\setminus \{v\}} 2^{-d_Q(v, x)} \\
   & = &  \sum_{x \in V(H_2)\setminus \{u_2\}} 2^{-d_{H_2}(u_2, x)} + \sum_{x \in V(H_1)\setminus \{v_1\}} 2^{-(d_{H_1}(v_1, x) +  d_{H_2}(v_2, u_2))} \\
   &\quad -&    \sum_{x \in V(H_2)\setminus \{v_2\}} 2^{-d_{H_2}(v_2, x)} -  \sum_{x \in V(H_1)\setminus \{v_1\}} 2^{-d_{H_1}(v_1, x)}\\
 &\leq  & \sum_{x \in V(H_1)\setminus \{v_1\}} 2^{-(d_{H_1}(v_1, x) +  d_{H_2}(v_2, u_2))}-\sum_{x \in V(H_1)\setminus \{v_1\}} 2^{-d_{H_1}(v_1, x)}<0.
\end{eqnarray*}
Hence, $C(G_2)< C(G)$.
\end{proof}

\begin{Lemma} \label{L1}
For $i=1,2, \dots, r$ and $r \geq 5$, let $H_i$ be a  graph with $u_i \in V(H_i)$.  Let $G$ be the graph obtained from $H_1, H_2, \dots, H_r$ and $C_r=x_1x_2 \cdots x_rx_1$ by identifying $u_i$ with $x_i$, $i=1,2, \dots, r$. Let $G' = G - u_1u_r +u_{r-2}u_r$. Then $C(G') < C(G)$.
\end{Lemma}
\begin{proof}
By the definition of closeness, we have
\begin{eqnarray*}\label{}
&&   C(G')- C(G)  \\
&=&\sum _{i=1}^{r} \sum_{x \in V(H_i) }\sum_{y \in V(H_i)\setminus \{x\}} \left(2^{-d_{G'}(x,y)}-2^{-d_G(x,y)}\right)\\
 &   + &      2\sum _{1\leq i\leq j \leq r}\left [\sum_{x \in V(H_i) \setminus \{u_i\}}\sum_{y \in V(H_j)\setminus \{u_j\}}  \left(2^{-d_{G'}(x,y)}-2^{-d_G(x,y)}\right)+  \left(2^{-d_{G'}(u_i, u_j)}-2^{-d_G(u_i, u_j)}\right)\right.\\
 &  + & \left.  \sum_{x \in V(H_i) \setminus \{u_i\}} \left (2^{-d_{G'}(x,u_j)}-2^{-d_G(x,u_j)}\right)+  \sum_{y \in V(H_j) \setminus \{u_j\}}  \left(2^{-d_{G'}(y,u_i)}-2^{-d_G(y, u_i)}\right) \right]\\
\end{eqnarray*}
It is easy to see that
$ d_G(x,y) =  d_{G'}(x,y)$,   for all $\{x, y \} \subseteq V(H_i)$ with $1 \leq i \leq r$.
 Therefore,
 \begin{eqnarray}\label{e}
  & & C(G')- C(G)  \nonumber \\
  & = &    2\sum _{1\leq i\leq j \leq r} \left [\sum_{x \in V(H_i) \setminus \{u_i\}}\sum_{y \in V(H_j)\setminus \{u_j\}}  2^{-(d_G(x,u_i) + d_G(u_j,y))} +  \sum_{x \in V(H_i) \setminus \{u_i\}} 2^{-d_{G}(x,u_i)} \right.\nonumber\\
   &   + & \left.        \sum_{y \in V(H_j) \setminus \{u_j\}}  2^{-d_{G'}(y,u_i)} + 1\right]\times \left(2^{-d_{G'}(u_i, u_j)}-2^{-d_G(u_i, u_j)}\right)\nonumber\\
                    & = &     2\sum _{1\leq i\leq j \leq r}\left [C(u_i)\cdot C(u_j)+ C(u_i)+ C(u_j)+1\right]\times \left(2^{-d_{G'}(u_i, u_j)}-2^{-d_G(u_i, u_j)}\right).
\end{eqnarray}
If $r$ is odd. Then it is easy to see that $d_{G}(u_i, u_j)= d_{G'}(u_i, u_j)$ for   $(i, j) \in V_1\cup V_2\cup V_3$, where $V_1= \{1 \leq i \leq  \frac{r-3}{2}, i+1 \leq j \leq  \frac{r-1}{2}+i\}$, $V_2=\{\frac{r-1}{2}\leq i \leq  r-2, i+1 \leq j \leq r-1\}$ and $V_3=\{i=r-1, j=r\}$.
From \eqref{e}, we have
 \begin{eqnarray*}\label{}
   &  & C(G')-   C(G)          \\
    & = & 2\left [ \sum _{i=1}^{\frac{r-3}{2}}\sum _{j= \frac{r+1}{2}+i}^{r} + \sum _{i=\frac{r-1}{2}}^{r-2}\sum _{j= r} \right]\left [C_{H_i}(u_i)\cdot C_{H_j}(u_j)+ C_{H_i}(u_i)+ C_{H_j}(u_j)+1\right]\\
    &  \times &    \left(2^{-d_{G'}(u_i, u_j)}-2^{-d_G(u_i, u_j)}\right)\\
    & = & 2\sum _{i=1}^{\frac{r-3}{2}}\sum _{j= \frac{r+1}{2}+i}^{r-1} \left [C_{H_i}(u_i)\cdot C_{H_j}(u_j)+ C_{H_i}(u_i)+ C_{H_j}(u_j)+1\right] \times \left(2^{-(j-i)}-2^{-(r-(j-i))}\right)\\
    &   + & 2\sum _{i=1}^{\frac{r-3}{2}} \left [C_{H_i}(u_i)\cdot C_{H_r}(u_r)+ C_{H_i}(u_i)+ C_{H_r}(u_r)+1\right] \times \left(2^{-(r-1-i)}-2^{-i}\right)\\
    &  + & 2\sum _{i= \frac{r+1}{2}}^{r-2} \left [C_{H_i}(u_i)\cdot C_{H_r}(u_r)+ C_{H_i}(u_i)+ C_{H_r}(u_r)+1\right] \times \left(2^{-(r-1-i)}-2^{-(r-i)}\right).\\
 \end{eqnarray*}
Note that $j-i > r-(j-i)$ for $\frac{r+1}{2}+i \leq j \leq r-1$.  We deduce that
 \begin{eqnarray*}\label{}
  && C(G')-   C(G)\\
      & < &   2\sum _{i=1}^{\frac{r-3}{2}} \left [C_{H_i}(u_i)\cdot C_{H_r}(u_r)+ C_{H_i}(u_i)+ C_{H_r}(u_r)+1\right]\left(2^{-(r-1-i)}-2^{-i}\right)\\
                           & \quad  + & 2\sum _{i= \frac{r+1}{2}}^{r-2} \left [C_{H_i}(u_i)\cdot C_{H_r}(u_r)+ C_{H_i}(u_i)+ C_{H_r}(u_r)+1\right]\left(2^{-(r-1-i)}-2^{-(r-i)}\right) \\
                  & =&   2\sum _{i=1}^{\frac{r-3}{2}} \left [C_{H_i}(u_i)\cdot C_{H_r}(u_r)+ C_{H_i}(u_i)+ C_{H_r}(u_r)+1\right]\left(2^{-(r-1-i)}-2^{-(i+1)}\right).
  \end{eqnarray*}
Note that $r-1-i > i+1$ for $1\leq i \leq \frac{r-3}{2}$, implying $2^{-(r-1-i)}-2^{-(i+1)}<0$.
Hence, $C(G') < C(G)$.

If $r$ is even. Then it is easy to see that $d_{G}(u_i, u_j)= d_{G'}(u_i, u_j)$ for   $(i, j) \in W_1\cup W_2\cup W_3$, where $W_1= \{1 \leq i \leq  \frac{r-4}{2}, i+1 \leq j \leq  \frac{r}{2}+i\}$, $W_2=\{\frac{r-2}{2}\leq i \leq  r-2, i+1 \leq j \leq r-1\}$ and $W_3=\{i=r-1, j=r\}$.
From \eqref{e}, we have
 \begin{eqnarray*}\label{}
   &  & C(G')-   C(G)          \\
    & = & 2\left [ \sum _{i=1}^{\frac{r-4}{2}}\sum _{j= \frac{r+2}{2}+i}^{r} + \sum _{i=\frac{r-2}{2}}^{r-2}\sum _{j= r} \right]\left [C_{H_i}(u_i)\cdot C_{H_j}(u_j)+ C_{H_i}(u_i)+ C_{H_j}(u_j)+1\right]\\
    &  \times &    \left(2^{-d_{G'}(u_i, u_j)}-2^{-d_G(u_i, u_j)}\right)\\
    & = & 2\sum _{i=1}^{\frac{r-4}{2}}\sum _{j= \frac{r+2}{2}+i}^{r-1} \left [C_{H_i}(u_i)\cdot C_{H_j}(u_j)+ C_{H_i}(u_i)+ C_{H_j}(u_j)+1\right] \times \left(2^{-(j-i)}-2^{-(r-(j-i))}\right)\\
    &   + & 2\sum _{i=1}^{\frac{r-4}{2}} \left [C_{H_i}(u_i)\cdot C_{H_r}(u_r)+ C_{H_i}(u_i)+ C_{H_r}(u_r)+1\right] \times \left(2^{-(r-1-i)}-2^{-i}\right)\\
    &   + & \left [C(u_{\frac{r-2}{2}})\cdot C_{H_r}(u_r)+ C(u_{\frac{r-2}{2}})+ C_{H_r}(u_r)+1\right] \times \left(2^{-(r-i)}-2^{-(r-1-i)}\right)\\
    &  + & 2\sum _{i= \frac{r}{2}}^{r-2} \left [C_{H_i}(u_i)\cdot C_{H_r}(u_r)+ C_{H_i}(u_i)+ C_{H_r}(u_r)+1\right] \times \left(2^{-(r-1-i)}-2^{-(r-i)}\right).
 \end{eqnarray*}
Note that $j-i > r-(j-i)$ for $\frac{r+2}{2}+i \leq j \leq r-1$.  We get
 \begin{eqnarray*}\label{}
  && C(G')-   C(G)\\
      & < &  2\sum _{i=1}^{\frac{r-4}{2}} \left [C_{H_i}(u_i)\cdot C_{H_r}(u_r)+ C_{H_i}(u_i)+ C_{H_r}(u_r)+1\right] \times \left(2^{-(r-1-i)}-2^{-i}\right)\\
    &   + & \left [C(u_{\frac{r-2}{2}})\cdot C_{H_r}(u_r)+ C(u_{\frac{r-2}{2}})+ C_{H_r}(u_r)+1\right] \times \left(2^{-(r-i)}-2^{-(r-1-i)}\right)\\
    &  + & 2\sum _{i= \frac{r}{2}}^{r-2} \left [C_{H_i}(u_i)\cdot C_{H_r}(u_r)+ C_{H_i}(u_i)+ C_{H_r}(u_r)+1\right] \times \left(2^{-(r-1-i)}-2^{-(r-i)}\right)\\
     & =&   2\sum _{i=1}^{\frac{r-2}{2}} \left [C_{H_i}(u_i)\cdot C_{H_r}(u_r)+ C_{H_i}(u_i)+C_{H_r}(u_r)+1\right]\left(2^{-(r-1-i)}-2^{-(i+1)}\right).
  \end{eqnarray*}
 Note  that $r-1-i > i+1$ for $1\leq i \leq \frac{r-2}{2}$, implying $2^{-(r-1-i)}-2^{-(i+1)}<0$.
Hence, $C(G') < C(G)$.
\end{proof}

\begin{Lemma} \label{L2}
Let $G$ be a  graph with a cycle $C_4=u_1u_2u_3u_4u_1$ such that $G-E(C_4)$ has exactly four components. Let $H_i$ be the component of $G-E(C_4)$ containing $u_i$, where $1 \leq i \leq 4$ and $|V(H_1)|= max_{1 \leq i \leq 4}\{|V(H_i)|\}$. Let $G' = G - u_1u_4 +u_2u_4$. Then $C(G') \leq C(G)$ with equality if and only if $|V(H_1)|=|V(H_2)|$.
\end{Lemma}
\begin{proof} By the definition of closeness, we get
\begin{eqnarray*}\label{}
 &&C(G')- C(G)   =  2\sum_{x \in V(H_1) \setminus \{u_1\}}\sum_{y \in V(H_4)\setminus \{u_4\}} \left(2^{-d_{G'}(x,y)}-2^{-d_G(x,y)}\right) \\
                &   + &   2\sum_{y \in V(H_4) \setminus \{u_4\}} \left(2^{-d_{G'}(u_1,y)}-2^{-d_G(u_1,y)}\right)+ 2\sum_{x \in V(H_1) \setminus \{u_1\}} \left(2^{-d_{G'}(u_4,x)}-2^{-d_G(u_4,x)}\right)\\
                &  + &  2\left(2^{-2}-2^{-1}\right) +2\sum_{x \in V(H_2) \setminus \{u_2\}}\sum_{y \in V(H_4)\setminus \{u_4\}} \left(2^{-d_{G'}(x,y)}-2^{-d_G(x,y)}\right)\\
                & + &   2\sum_{y \in V(H_4) \setminus \{u_4\}} \left(2^{-d_{G'}(u_2,y)}-2^{-d_G(u_2,y)}\right)+ 2\sum_{x \in V(H_2) \setminus \{u_2\}} \left(2^{-d_{G'}(u_4,x)}-2^{-d_G(u_4,x)}\right)\\
                &  + &   2(2^{-1}-2^{-2}).
 \end{eqnarray*}
Note that $d_{G'}(x, u_i)=d_{G}(x, u_i)$ for all $x \in V(H_i)\setminus \{u_i\}$ with $1 \leq i \leq r$. Therefore, we get
\begin{eqnarray*}\label{}
 &&C(G')- C(G)     \\
 &=& 2\left [ \sum_{x \in V(H_1) \setminus \{u_1\}}\sum_{y \in V(H_4)\setminus \{u_4\}} 2^{-(d_{G}(x,u_1)+ d_{G}(u_4,y))}   +\sum_{y \in V(H_4) \setminus \{u_4\}} 2^{-d_{G}(u_4,y)}\right.\\
 &+& \left. \sum_{x \in V(H_1) \setminus \{u_1\}}2^{-d_G(u_1,x)} \right]\times(2^{-2}-2^{-1})+ 2 \left [ \sum_{x \in V(H_2) \setminus \{u_2\}}\sum_{y \in V(H_4)\setminus \{u_4\}} 2^{-(d_{G}(x,u_2) + d_{G}(u_4,y))}\right.\\
 &  + & \left. \sum_{y \in V(H_4) \setminus \{u_4\}} 2^{-d_G(u_4,y)} + \sum_{x \in V(H_2) \setminus \{u_2\}} 2^{-d_{G}(u_2,x)}\right]\times(2^{-1}-2^{-2})\\
 & =  & 2\left[C_{H_1}(u_1)C_{H_4}(u_4) + C_{H_1}(u_1) + C_{H_4}(u_4)\right]\times(2^{-2}-2^{-1})\\
  &  + & 2\left[C_{H_2}(u_2)C_{H_4}(u_4) + C_{H_2}(u_2) + C_{H_4}(u_4)\right] \times (2^{-1}-2^{-2})\\
& =  & (C_{H_4}(u_4)+1)\left[C_{H_1}(u_1) - C_{H_2}(u_2)\right](-2^{-1})\\
& \leq & (C_{H_4}(u_4)+1)\left(|V(H_1)|-|V(H_2)|\right)(-2^{-2})\leq 0
\end{eqnarray*}
with equality if and only if $|V(H_1)|=|V(H_2)|$.
\end{proof}

\begin{Lemma} \label{L4}
Let $G$ be a cactus with a cycle $C_3=u_1u_2u_3u_1$ such that $G-E(C_3)$ has exactly three components. Let $H_i$ be the end-blocks of $G-E(C_3)$ containing $u_i$.  Let $v_1$ (resp. $v_2$) be the farthest vertex from $u_1$ (resp. $u_2$) in $H_1$ (resp. $H_2$) satisfying $C_{H_i}(v_i) \leq min_{1\leq i\leq 2}\{C_{H_i}(u_i)\}$. Let $G_1= G- \sum _{w \in N_{H_3}(u_3)}wu_3 + \sum _{w \in N_{H_3}(u_3)}wv_1$ and $G_2= G- \sum _{w \in N_{H_3}(u_3)}wu_3 + \sum _{w \in N_{H_3}(u_3)}wv_2$. Then either $C(G_1) < C(G)$ or $  C(G_2) < C(G) $.
\end{Lemma}
\begin{proof}
Let $M= G - V(H_3) \setminus \{u_3\}$, $M_1= G_1 - V(H_3) \setminus \{u_1\}$ and $M_2= G_2 - V(H_3) \setminus \{u_2\}$. One has $M \cong M_1 \cong M_2$ and $C_{H_3}(u_3)=C_{H_3}(v_1)=C_{H_3}(v_2)$.
By Lemma \ref{111}, we have
\begin{equation*}\label{}
C(G)= C(H_3)+C(M)+2C_{H_3}(u_3)C_{M}(u_3),
\end{equation*}
\begin{equation*}\label{}
C(G_1)= C(H_3)+C(M)+2C_{H_3}(u_3)C_{M}(v_1),
\end{equation*}
and
\begin{equation*}\label{}
C(G_2)= C(H_3)+C(M)+2C_{H_3}(u_3)C_{M}(v_2).
\end{equation*}
This gives
\begin{equation}\label{e11}
C(G_1)- C(G) = 2C_{H_3}(u_3)\left[C_{M}(v_1)- C_{M}(u_3)\right],
\end{equation}
and
\begin{equation}\label{e12}
C(G_2)- C(G) = 2C_{H_3}(u_3)\left[C_{M}(v_2)- C_{M}(u_3)\right].
\end{equation}
If $C_{H_1}(u_1) \leq C_{H_2}(u_2)$. Then from \eqref{e11}, we get
 \begin{eqnarray*}\label{}
 &  & C_{M}(v_1)- C_{M}(u_3)   =   \sum_{x \in V(M)\setminus \{v_1\}} 2^{-d_M(v_1, x)} - \sum_{x \in V(M)\setminus \{u_3\}} 2^{-d_M(u_3,x)} \\
  & = &  \sum_{x \in V(H_1)\setminus \{v_1\}} 2^{-d_{H_1}(v_1, x)} + 2^{-(d_{H_1}(v_1, u_1)+1)}+ \sum_{x \in V(H_2)\setminus \{u_2\}} 2^{-(d_{H_2}(u_2,x)+1 +  d_{H_1}(v_1,u_1))}\\
  & \quad + &           2^{-(d_{H_1}(v_1, u_1)+1)}   -\sum_{x \in V(H_1)\setminus \{u_1\}} 2^{-(d_{H_1}(u_1, x)+1)} -2^{-1} - \sum_{x \in V(H_2)\setminus \{u_2\}} 2^{-(d_{H_2}(u_2,x)+1)}-2^{-1}\\
   & \leq &  2^{-d_{H_2}(v_2, u_2)} -2^{-1}+\sum_{x \in V(H_2)\setminus \{u_2\}} 2^{-(d_{H_2}(u_2,x)+1 +  d_{H_2}(v_2,u_2))}-2^{-1}\\
   & \leq &  2^{-d_{H_2}(v_2, u_2)} -2^{-1}+|V(H_2)-1| [2^{-(d_{H_2}(v_2,u_2)+2)}]-2^{-1}.
  \end{eqnarray*}
Note that $d_{H_2}(v_2, u_2) \geq 1$ and $G$  is a cactus, so  $H_2$ must be either a cycle or a cut edge,  $|V(H_2)-1| \leq 2d_{H_2}(v_2, u_2)$. Therefore, we get
 $C_{A}(v_1)- C_{A}(u_3)  \leq 2 \times 2^{-3}-2^{-1}<0$. Hence, $C(G_1)< C(G)$.

If $C_{H_1}(u_1) > C_{H_2}(u_2)$. Then by \eqref{e12}, we have
 \begin{eqnarray*}\label{}
 &  & C_{M}(v_2)- C_{M}(u_3)   =   \sum_{x \in V(M)\setminus \{v_2\}} 2^{-d_M(v_2, x)} - \sum_{x \in V(M)\setminus \{u_3\}} 2^{-d_M(u_3,x)} \\
  & = &  \sum_{x \in V(H_2)\setminus \{v_2\}} 2^{-d_{H_2}(v_2, x)} + 2^{-(d_{H_2}(v_2, u_2)+1)}+ \sum_{x \in V(H_1)\setminus \{u_1\}} 2^{-(d_{H_1}(u_1,x)+1 +  d_{H_2}(v_2,u_2))}\\
  & \quad + &           2^{-(d_{H_2}(v_2, u_2)+1)}   - \sum_{x \in V(H_2)\setminus \{u_2\}} 2^{-(d_{H_2}(u_2, x)+1)} -2^{-1} - \sum_{x \in V(H_1)\setminus \{u_1\}} 2^{-(d_{H_1}(u_1,x)+1)}-2^{-1}\\
   & < &  2^{-d_{H_1}(v_1, u_1)} - 2^{-1}+\sum_{x \in V(H_1)\setminus \{u_1\}} 2^{-(d_{H_1}(u_1, x)+1 +  d_{H_1}(v_1, u_1))}-2^{-1}\\
   & \leq &  2^{-d_{H_1}(v_1, u_1)} -2^{-1}+|V(H_1)-1| [2^{-(d_{H_1}(v_1,u_1)+2)}]-2^{-1}.
  \end{eqnarray*}
Note that $d_{H_1}(v_1, u_1) \geq 1$ and $G$  is a cactus, so  $H_1$ must be either a cycle or a cut edge,  we have  $|V(H_1)-1| \leq 2d_{H_1}(v_1, u_1)$. Therefore, we get
$ C_{M}(v_2)- C_{M}(u_3)   <   2 \times 2^{-3}-2^{-1}<0.$
Hence, $C(G_2)< C(G)$.
\end{proof}

\begin{Lemma} \label{L5}
 Let $H_1, H_2$ be two connected graphs such that $H_2$ is not a path, $u \in V(H_1)$ and $|V(H_1)| \geq 2$. Let $v_1, v_2$ be the end-vertices of a longest path of $H_2$. Clearly, $|V(H_2)| \geq 3$. Let $G$ (resp. $G_1$) (resp. $G_2$) be the graphs obtained from  $H_1, H_2$ and $P_r = w_1w_2 \cdots w_r$ by identifying $u$ with $v_1$ and $v_2$ with $w_1$  (resp. $u$ with $w_r$ and $v_2$ with $w_1$) (resp. $u$ with $w_r$ and $v_1$ with $w_1$).

 $(i)$ If $C_{H_2}(v_2) \leq C_{H_2}(v_1)$, then $C(G_1) < C(G)$.

 $(ii)$ If $C_{H_2}(v_1) > C_{H_2}(v_2)$, then $C(G_2) < C(G)$.
 \end{Lemma}
\begin{proof}
Let $A= G - V(H_1) \setminus \{u\}$, $A_1= G_1 - V(H_1) \setminus \{w_r\}$ and $A_2= G_2 - V(H_1) \setminus \{w_r\}$. One has $A \cong A_1$. By Lemma \ref{111}, we have
\begin{equation}\label{e21}
C(G)= C(H_1) + C(A) + 2C_{H_1}(u)C_{A}(v_1),
\end{equation}
\begin{equation}\label{e22}
C(G_1)= C(H_1) + C(A) + 2C_{H_1}(u)C_{A}(w_r),
\end{equation}
\begin{equation}\label{e23}
C(G_2)= C(H_1) + C(A_2) + 2C_{H_1}(u)C_{A_2}(w_r).
\end{equation}

$(i)$ If $C_{H_2}(v_2) \leq C_{H_2}(v_1)$. Then by \eqref{e21} and \eqref{e22}, we get
\begin{equation}\label{e24}
   C(G_1) -   C(G)    =  2C_{H_1}(u) \left[C_{A}(w_r) - C_{A}(v_1) \right].
  \end{equation}
From \eqref{e24}, we have
\begin{eqnarray*}\label{}
& &C_{A}(w_r) - C_{A}(v_1) =  \sum_{x \in V(A)\setminus \{w_r\}} 2^{-d_A(x, w_r)}-\sum_{x \in V(A)\setminus \{v_1\}} 2^{-d_A(x, v_1)}\\
& = & \sum _{i= 1}^{r-1}2^{-d_{P_r}(w_i, w_r)} + \sum_{x \in V(H_2)\setminus \{v_2\}} 2^{-[ d_{H_2}(x, v_2) +  d_{P_r}(w_1, w_r)]}-\sum_{x \in V(H_2)\setminus \{v_1\}} 2^{-d_{H_2}(x, v_1)}\\
    &\qquad - & \sum _{i= 2}^{r} 2^{-[ d_{H_2}(v_1, v_2) +  d_{P_r}(w_1, w_i)]}\\
& = & \left[1-2^{-(r-1)}\right] + \sum_{x \in V(H_2)\setminus \{v_2\}}2^{- d_{H_2}(x, v_2) }-\sum_{x \in V(H_2)\setminus \{v_2\}}2^{- d_{H_2}(x, v_2) }\\
&\qquad + &  \sum_{x \in V(H_2)\setminus \{v_2\}} 2^{-[ d_{H_2}(x, v_2) +  r-1]}-\sum_{x \in V(H_2)\setminus \{v_1\}} 2^{-d_{H_2}(x, v_1)} - \left[1-2^{-(r-1)}\right]\\
&\qquad + & \left[1-2^{-(r-1)}\right] -  \left[1-2^{-(r-1)}\right]2^{- d_{H_2}(v_1, v_2)}\\
& = & C_{H_2}(v_2) - C_{H_2}(v_1)-  \left[1-2^{-(r-1)}\right]C_{H_2}(v_2) + \left[1-2^{-(r-1)}\right]\left[1-2^{- d_{H_2}(v_1, v_2)}\right]\\
& \leq &   \left[1-2^{-(r-1)}\right]\left[1-2^{- d_{H_2}(v_1, v_2)}-C_{H_2}(v_2)\right].
\end{eqnarray*}
Note that  $|V(H_2)| \geq 3$  and $H_2$ is not a path. Then there exist at least two vertices, say $x_1, x_2 \in V(H_2)$ such that  $d_{H_2}(x_1, v_2)=d_{H_2}(x_2, v_2)=1$. Therefore, $C_{H_2}(v_2) >1$, i.e., $\left[1-2^{- d_{H_2}(v_1, v_2)}-C_{H_2}(v_2)\right] < 0$. Hence, $ C(G_1) <  C(G) $.

$(ii)$ If $C_{H_2}(v_2) > C_{H_2}(v_1)$. Then by \eqref{e21} and \eqref{e23}, we get
\begin{equation}\label{e25}
   C(G_2) -   C(G)    = C(A_2)-C(A) +  2C_{H_1}(u) \left[C_{A_2}(w_r) - C_{A}(v_1) \right].
  \end{equation}
By Lemma \ref{111}, we have
\begin{equation*}\label{}
   C(A_2)     = C(H_2)+C(P_r) +  2C_{H_2}(v_1) C_{P_r}(w_1),
  \end{equation*}
  \begin{equation*}\label{}
   C(A)     = C(H_2)+C(P_r) +  2C_{H_2}(v_2) C_{P_r}(w_1).
  \end{equation*}
This gives
\begin{equation}\label{e26}
   C(A_2) - C(A)    =  2C_{P_r}(w_1)\left[C_{H_2}(v_1) - C_{H_2}(v_2) \right]<0.
  \end{equation}
 Also,  one has
  \begin{eqnarray}\label{e27}
& &C_{A_2}(w_r) - C_{A}(v_1) =  \sum_{x \in V(A_2)\setminus \{w_r\}} 2^{-d_{A_2}(x, w_r)}-\sum_{x \in V(A)\setminus \{v_1\}} 2^{-d_A(x, v_1)}\nonumber\\
& = & \sum _{i= 1}^{r-1}2^{-d_{P_r}(w_i, w_r)} + \sum_{x \in V(H_2)\setminus \{v_1\}} 2^{-[ d_{H_2}(x, v_1) +  d_{P_r}(w_1, w_r)]}-\sum_{x \in V(H_2)\setminus \{v_1\}} 2^{-d_{H_2}(x, v_1)}\nonumber\\
    &\qquad - & \sum _{i= 2}^{r} 2^{-[ d_{H_2}(v_1, v_2) +  d_{P_r}(w_1, w_i)]}\nonumber\\
& = & \left[1-2^{-(r-1)}\right] + \sum_{x \in V(H_2)\setminus \{v_1\}}2^{- d_{H_2}(x, v_1) }-\sum_{x \in V(H_2)\setminus \{v_1\}}2^{- d_{H_2}(x, v_1) }\nonumber\\
&\qquad + &  \sum_{x \in V(H_2)\setminus \{v_1\}} 2^{-[ d_{H_2}(x, v_1) +  r-1]}-\sum_{x \in V(H_2)\setminus \{v_1\}} 2^{-d_{H_2}(x, v_1)} - \left[1-2^{-(r-1)}\right]\nonumber\\
&\qquad + & \left[1-2^{-(r-1)}\right] -  \left[1-2^{-(r-1)}\right]2^{- d_{H_2}(v_1, v_2)}\nonumber\\
& = &   \left[1-2^{-(r-1)}\right]\left[1-2^{- d_{H_2}(v_1, v_2)}-C_{H_2}(v_1)\right].
\end{eqnarray}
Note that  $|V(H_2)| \geq 3$. Then there exist at least two vertices, say $s_1, s_2 \in V(H_2)$ such that  $d_{H_2}(s_1, v_1)=d_{H_2}(s_2, v_1)=1$. Therefore, $C_{H_2}(v_1) >1$, i.e., $\left[1-2^{- d_{H_2}(v_1, v_2)}-C_{H_2}(v_1)\right] < 0$. Thus, $ C_{A_2}(w_r) - C_{A}(v_1) < 0$.

Combining \eqref{e25}, \eqref{e26} and \eqref{e27}, we have,  $ C(G_2) <  C(G)$, as desired.
\end{proof}

\begin{Lemma} \label{L6}
 Let $H_1, H_2$ be two connected graphs with $v_1 \in V(H_1)$ and $v_2 \in V(H_2)$. Without loss of generality, we assume that   $|V(H_1)| \leq  |V(H_2)|$. Let $P_r = w_1w_2 \cdots w_r$ with $r \geq 4$. Put $M := P_r + w_{r-2}w_r$ and $M' := M - \{w_{r-2}w_r, w_{r-1}w_r\} + \{w_1w_r, w_2w_r\}$.   Let $G_1$ (resp. $G_2$)  be the graphs obtained from  $H_1, H_2$ and $M$ (resp. $M'$) by identifying  $v_1$ with $w_1$ and $v_2$ with $w_{r-1}$  (resp. $v_1$ with $w_1$ and $v_2$ with $w_{r-1}$). Then  $C(G_1) \leq C(G_2)$ with equality if and only if $|V(H_1)| =  |V(H_2)|$.
 \end{Lemma}
\begin{proof}
Let  $Z_1= G_1 - V(H_2) \setminus \{v_2\}$ and $Z_2= G_2 - V(H_2) \setminus \{v_2\}$. By Lemma \ref{111}, we have
\begin{equation*}\label{}
C(G_1)= C(H_2) + C(Z_1) + 2C_{H_2}(v_2)C_{Z_1}(v_2),
\end{equation*}
\begin{equation*}\label{}
C(G_2)= C(H_2) + C(Z_2) + 2C_{H_2}(v_2)C_{Z_2}(v_2).
\end{equation*}
From the above, we get
\begin{equation}\label{e28}
   C(G_2) -   C(G_1)    = C(Z_2)-C(Z_1)+ 2C_{H_2}(v_2) \left[C_{Z_2}(v_2) - C_{Z_1}(v_2) \right].
  \end{equation}
  By Lemma \ref{111}, we have
\begin{equation*}\label{}
   C(Z_2)     = C(H_1)+C(M') +  2C_{H_1}(v_1) C_{M'}(w_1),
  \end{equation*}
  \begin{equation*}\label{}
   C(Z_1)     = C(H_1)+C(M) +  2C_{H_1}(v_1) C_{M}(w_1).
  \end{equation*}
Note that $M \cong M'$, we have
 \begin{eqnarray}\label{e29}
   C(Z_2) - C(Z_1)   & = & 2C_{H_1}(v_1)\left[ C_{M'}(w_1) - C_{M}(w_1) \right] \nonumber\\
                      & =& 2C_{H_1}(v_1)\left[ \sum _{i= 2}^{r-1}2^{-(i-1)}+2^{-1} - \sum _{i= 2}^{r-1}2^{-(i-1)}-2^{-(r-2)} \right] \nonumber\\
                      & =& 2C_{H_1}(v_1)\left[ 2^{-1} -2^{-(r-2)} \right]
  \end{eqnarray}
   Moreover,  one has
  \begin{eqnarray}\label{e30}
C_{Z_2}(v_2) - C_{Z_1}(v_2) & = & \sum_{x \in V(Z_2)\setminus \{v_2\}} 2^{-d_{Z_2}(x, v_2)}-\sum_{x \in V(Z_1)\setminus \{v_2\}} 2^{-d_Z(x, v_2)}\nonumber\\
& = & \sum _{i= 2}^{r-1}2^{-(r-1-i)}+ 2^{-(r-2)} + \sum_{x \in V(H_1)\setminus \{v_1\}} 2^{-[d_{H_1}(x, v_1)+r-2]}\nonumber\\
&\qquad - & \sum _{i= 2}^{r-1}2^{-(r-1-i)}- 2^{-1} - \sum_{x \in V(H_1)\setminus \{v_1\}} 2^{-[d_{H_1}(x, v_1)+r-2]}\nonumber\\
 & =& \left[ 2^{-(r-2)}- 2^{-1} \right].
\end{eqnarray}
 Combining \eqref{e28}, \eqref{e29} and \eqref{e30}, we get
 \begin{eqnarray*}\label{}
   C(G_2) -   C(G_1)    &=& 2C_{H_1}(v_1)\left[ 2^{-1} -2^{-(r-2)} \right]+ 2C_{H_2}(v_2) \left[2^{-(r-2)}- 2^{-1} \right]\\
   &=& 2\left[C_{H_1}(v_1) -  C_{H_2}(v_2) \right] \left[2^{-1}-2^{-(r-2)} \right]\\
    &\leq& \left(|V(H_1)| -  |V(H_2)| \right) \left[2^{-2}-2^{-(r-1)} \right] \leq 0.
  \end{eqnarray*}
with  equality if and only if $|V(H_1)| =  |V(H_2)|$.
\end{proof}

\section{Proof of Theorem \ref{a}}
Now we are ready to give a proof to Theorem \ref{a}

\begin{proof}
Let $G \in \mathcal{D}(n, k)$ such that $C(G)$ is as small as possible.

If $k =0$, then $G$ is a tree. If $G$ has more than two pendent vertices, then  applying Lemma \ref{L3} repeatedly to $G$, we get $G \cong D(n; 0, 0)=P_n$.

In what follows, we suppose that $k\geq 1$.
Let $r$ be the length of the cycle contained in $G$.  We claim that $r =3$.
Suppose that $r \geq 4$. If $r \geq 5$, then by Lemma \ref{L1}, we obtain another graph $G' \in \mathcal{D}(n, k)$ such that $C(G')<C(G)$, a contradiction. If $r = 4$, i.e., $G$ contains a cycle $C_4=u_1u_2u_3u_4u_1$ such that $G-E(C_4)$ has exactly four components. Let $H_i$ be the component of $G-E(C_4)$ containing $u_i$, where $1 \leq i \leq 4$ and $|V(H_1)|= max_{1 \leq i \leq 4}\{|V(H_i)|\}$. If $|V(H_1)| \neq |V(H_2)|$, then according to Lemma \ref{L2}, we get $G' \in \mathcal{D}(n, k)$ such that $C(G')<C(G)$, a contradiction. If $|V(H_1)| = |V(H_2)|$, then by Lemma \ref{L2}, we have $C(G')=C(G)$. By Lemma \ref{L3}, there exists $G'' \in \mathcal{D}(n, k)$ such that $C(G'')<C(G')$, which follows $C(G'')<C(G)$, a contradiction to the choice of $G$. Thus $r =3$, as desired.

By Lemma \ref{L4}, there exist at most two vertices with degrees greater than 2 in each triangle of $G$. And by Lemma \ref{L3},  no three blocks are having a common vertex in $G$. Thus, $G$ is a cactus obtained from $P_{n-k}=u_1u_2 \cdots u_{n-k}$ by adding $k$ vertices $\{v_1, v_2, \dots, v_k\}$ and $2k$ edges such that $N_G(v_x)= \{u_i, u_{i+1}\}$ for some $1 \leq i \leq n-k-1$ and $|N_G(v_x)\cap N_G(v_y)| \leq 1$ for any $1 \leq x < y \leq k$.

If $k=1$, then we claim that $G$ has exactly one pendent path. Otherwise, by Lemma \ref{L5}, we obtain another graph $G' \in \mathcal{D}(n, k)$
 such that $C(G')<C(G)$, a contradiction. Thus,  $G$ has exactly one pendent path.

If $k \geq 2$. Then by Lemma \ref{L5}, $G$ has no pendent path. We claim that  $G$ has at most one internal path. Suppose that $G$ has two internal paths, say $P, Q$. Since $G$ has no pendent path, and each block of $G$ is either a cycle or a cut edge. Therefore, each of the end-block of $G$ is a triangle. We may partition  $G$ as $G_1\cup G_0 \cup G_2$, where $G_1$ is composed of the leftmost successive triangles of length, say $\ell_1$, $G_0$ is the path $P$ together with triangle(s) (that lie between the two paths $P$ and $Q$), and $G_2$ is composed of the path $Q$ together with the rightmost successive triangles of length, say $\ell_2$. Without loss of generality, we assume that  $\ell_2 \geq \ell_1 \geq 1 $. Clearly, $|V(G_1)|= 2 \ell_1 +1 < 2 \ell_2 +2 \leq  |V(G_2)|$. By Lemma \ref{L6}, there exists another graph $G' \in \mathcal{D}(n, k)$ such that $C(G')<C(G)$, a contradiction to the choice of $G$. Thus, $G$ has at most one internal path, i.e., $G \cong D(n; k_1, k_2)$, where $k_1 + k_2 = k$.

We claim that $|k_1 - k_2| \leq 1$. Otherwise, by Lemma \ref{L6}, there exists another graph $G' \in \mathcal{D}(n, k)$ such that $C(G')<C(G)$, a contradiction. Thus, $|k_1 - k_2| \leq 1$, i.e., $G \cong D(n; \lfloor \frac{k}{2} \rfloor, \lceil\frac{k}{2}\rceil)$.
\end{proof}

\section{Concluding remarks}
Poklukar and \v{Z}erovnik \cite{PZ} determined the unique graph that maximizes the closeness among all cacti with a fixed number of vertices and cycles and posed an open problem regarding the minimum case. In  this article, we solved  the open  problem proposed by Poklukar and \v{Z}erovnik. We obtained the graph that uniquely minimizes the closeness among all cacti with a fixed number of vertices and cycles, which completed the characterization of the graphs with extremal (maximum and minimum) closeness over all cacti with a fixed number of vertices and cycles. For future work, it will be  interesting to determine the graphs with the second maximum and minimum closeness in the class of $n$-vertex cacti with a fixed number of cycles.

\vspace{3mm}
\noindent {\bf Acknowledgement:} This work was supported  by National Natural Science Foundation of China (Grant Nos. 12071194, 11571155).

\end{document}